\documentclass[acmtocl]{acmtrans2m}

\usepackage{times}
\usepackage{amsmath}
\usepackage{amssymb}

\newtheorem{thm}{Theorem}[section]
\newtheorem{cor}{Corollary}[section]
\newtheorem{lem}{Lemma}[section]

\newtheorem{defn}{Definition}[section]
\newtheorem{prop}{Proposition}[section]
\newtheorem{exam}{Example}[section]

\def\>{\ensuremath{\rangle}}
\def\<{\ensuremath{\langle}}

\acmVolume{2} \acmNumber{3} \acmYear{01} \acmMonth{09}

\newcommand{\BibTeX}{{\rm B\kern-.05em{\sc i\kern-.025em b}\kern-.08em
    T\kern-.1667em\lower.7ex\hbox{E}\kern-.125emX}}

\title{Model-Checking Linear-Time Properties of Quantum Systems}

\author{MINGSHENG YING, YANGJIA LI, NENGKUN YU, and YUAN FENG\\University of Technology, Sydney and Tsinghua University}

\begin{abstract}
We define a formal framework for reasoning about linear-time properties of quantum systems in which quantum automata are employed in the modeling of systems and certain closed subspaces 
of state (Hilbert) spaces are used as the atomic propositions about the behavior of systems. We provide an algorithm for verifying invariants of quantum automata. Then automata-based model-checking technique is generalized for the verification of safety properties recognizable by reversible automata and $\omega-$properties recognizable by reversible B\"uchi automata.   
\end{abstract}





\begin{document}

\setcounter{page}{1}

\begin{bottomstuff}
This work was partly supported by the Australian Research Council (Grant No: 
DP110103473) and 
the National Natural Science
Foundation of China (Grant No: 60736011)
\newline
Authors' address: Center of Quantum Computation and
Intelligent Systems, Faculty of Engineering and Information
Technology, University of Technology, Sydney, City Campus, 15
Broadway, Ultimo, NSW 2007, Australia, and State Key Laboratory of
Intelligent Technology and Systems, Tsinghua National Laboratory for
Information Science and Technology, Department of Computer Science
and Technology, Tsinghua University, Beijing 100084, China, email:
mying@it.uts.edu.au or yingmsh@tsinghua.edu.cn
\end{bottomstuff}

\maketitle

\section{Introduction}

\subsection{Quantum Engineering} 
As pointed out by Dowling and Milburn~\cite{DM03}, we are currently
in the midst of a second quantum revolution: transition from quantum
theory to quantum engineering. The aim of quantum theory is to find
fundamental rules that govern the physical systems already existing
in the nature. Instead, quantum engineering intends to design and
implement new systems (machines, devices, etc) that do not exist
before to accomplish some desirable tasks, based on quantum theory.
Experiences in today's engineering indicate that it is not
guaranteed that a human designer completely understands the
behaviors of the systems she/he designed, and a bug in her/his
design may cause some serious problems and even disasters. So,
correctness, safety and reliability of complex engineering systems
have attracted wide attention and have been systematically studied in
various engineering fields. As is well-known, human intuition is
much better adapted to the classical world than the quantum world.
This implies that human engineers will commit many more faults in
designing and implementing complex quantum systems. Thus,
correctness, safety and reliability problem will be even more
critical in quantum engineering than in today's engineering.

\subsection{Model-Checking} In the last four decades, computer scientists have systematically
developed theories of correctness and safety as well as
methodologies, techniques and even automatic tools for correctness
and safety verification of computer systems; see for
example~\cite{La77}, \cite{MP95}, \cite{AS85}. Model-checking is an effective
automated technique that checks whether a formal (temporal logic)
property is satisfied in a formal model of a system. It has become
one of the dominant techniques for verification of computer systems
nearly 30 years after its inception. Many industrial-strength
systems have been verified by employing model-checking techniques.
Recently, it has also successfully been used in systems biology; see~\cite{HKNPT06} for example. 

\subsection{Model-Checking Quantum Systems}\label{probs}  A question then naturally arises: is it possible and how to
use model-checking techniques to verify correctness and safety of quantum engineering systems?
It seems that the current model-checking
techniques cannot be directly applied to quantum systems because of
some essential differences between the classical world and the quantum world. To develop model-checking techniques for quantum systems, at least the following two problems must be addressed: 
\begin{itemize}\item 
The classical system modeling method cannot be used to describe the behaviors of quantum systems, and the classical specification language is not suited to formalize the properties of quantum systems to be checked.      
So, we need to
carefully and clearly define a conceptual framework in which we can
properly reason about quantum systems, including formal models of
quantum systems and formal description of temporal properties of
quantum systems.\item The state spaces of the classical systems that model-checking techniques can be applied to are usually finite or countably infinite. However, the state spaces of quantum systems are inherently continuous even when they are finite-dimensional. In order to check quantum systems, we have to exploit some deep mathematical properties so that it suffices to examine only a finite number of (or at most countably infinitely many) representative elements, e.g. those in an orthonormal basis, of their state spaces.\end{itemize} 

\subsection{Previous Works}

There have been quite a few papers devoted to model-checking quantum systems. Almost all of the previous works target checking quantum communication protocols. For example, Gay, Nagarajan and Papanikolaou~\cite{GNP06} used the probabilistic model-checker PRISM~\cite{KNP04} to verify the correctness of several quantum protocols including BB84~\cite{BB84}. 
Furthermore, they~\cite{GNP08}, \cite{Pa08} developed an automatic tool QMC (Quantum Model-Checker). QMC uses the stabilizer formalism~\cite{Go97} for the modeling of systems, and the properties to be checked by QMC are expressed in Baltazar, Chadha, Mateus and Sernadas' quantum computation tree logic~\cite{BCMS07}, \cite{BCM08}. But as we shall see below, both the motivations and approaches of the works mentioned are very different from those of this paper.    

There are other two related research lines of verifying the correctness of quantum systems in the previous literature:  (1) quantum process algebras~\cite{GN05}, \cite{JL04}, \cite{FDJY07}, \cite{YFDJ09}, \cite{FDY11}, \cite{YF09}, and (2) quantum simulation, see~\cite{Ll96} for example. They are pursued by computer scientists and physicists, respectively. All works in these two lines have not employed model-checking techniques. 

\subsection{Design Decision of the Paper} Our purpose is to develop model-checking techniques that can be used not only for quantum communication protocols but also for other quantum engineering systems. To this end, first of all, we must address the first problem raised in Subsec.~\ref{probs}. This paper is merely one of the first steps toward such a general purpose. So, we choose to consider a simple formal model as well as a class of simple properties of quantum systems to be checked. More precisely, the major design decision of this paper is as follows:     
\begin{itemize}\item A quantum automaton~\cite{KW97} is adopted as the model of the system. This is obviously reasonable since classical automata (or equivalently transition systems) are the common system models in classical model-checking. 

\item Only linear-time properties of quantum systems are checked in this paper. They are defined to be infinite sequences of sets of atomic propositions, as in the classical case. But atomic propositions about quantum systems are essentially different from those for classical systems. Certain closed subspaces of the state (Hilbert) space of the system are chosen as atomic propositions about the system. The idea of viewing closed subspaces of (equivalently, projections on) a Hilbert space as propositions about a quantum system can be traced back to Birkhoff and von Neumann~\cite{BvN36}, and has been widely accepted in the quantum logic community for more than 70 years. 
\end{itemize}

\subsection{Contribution of the paper} Overall, automata-based model-checking techniques~\cite{VW94}, \cite{KV01} are generalized into the quantum setting. The key idea of the automata-based approach to model-checking is that we can use an auxiliary automaton to recognize the properties to be checked, and then it is combined with the system under checking so that the problem of checking the safety or $\omega-$properties of the system is reduced to checking some simpler (invariant or persistence) properties of the larger system composed by the system under checking and the auxiliary automaton. A difference between the classical case and the quantum case deserves a careful explanation. In the classical case, the auxiliary automaton can be any finite state automaton, whereas in the quantum case, such an auxiliary automaton is required to be reversible; otherwise it cannot be a part of a quantum system because the dynamics of a quantum system is inherently reversible. Since some regular and $\omega-$regular languages cannot be recognized by reversible automata~\cite{P87}, \cite{P01}, the class of properties that can be checked by the techniques developed in this paper is a proper subclass of that by classical model-checking techniques (if we ignore the difference between classical and quantum atomic propositions).    

The major technical contribution of this paper is a solution to the second problem raised in Subsec.~\ref{probs}. This solution consists of the following steps:
\begin{enumerate} \item Under an assumption about commutativity of atomic propositions, we show that to check an invariant of a quantum automaton, it suffices to examine its behaviors starting in an orthonormal basis of the space of its initial states. Thus, an algorithm for checking invariants of quantum automata can be developed since there are only a finite number of elements in a basis of a finite-dimensional state space.
\item Under the same assumption, it is shown that a quantum automaton satisfies a persistence property if and only if it satisfies a corresponding invariant. This is very different from the classical case, and at the first glance it is quite strange. However, such an equivalence between invariants and persistence properties is reasonable because the operations of quantum automata are always reversible. \item We show that the reduction from safety and $\omega-$properties of the system under checking to invariants and persistence properties of the composed system stated above is feasible if the composed system always starts in an orthonormal basis of the space of its initial states. 
\item Fortunately, we can choose a set of atomic propositions about the composed system that enjoys the required commutativity. This enables us to connect 1), 2) and 3) seamlessly. It is worth noting that one of the main technical difficulties in quantum model-checking is to find a way in which such a connection is effective. Indeed, this connection heavily depends on some profound properties of Hilbert spaces, e.g. implication from commutativity to distributivity in the lattice of closed subspaces of a Hilbert space. However, this connection works automatically and so was not a problem at all in the classical case. 
 \end{enumerate}

\subsection{Organization of the Paper} In Sec.~\ref{frame}, we recall some basic notions from quantum theory as well as the definition of quantum automata from~\cite{KW97} for convenience of the reader. 
In Sec.~\ref{ltp}, a language for specifying linear-time properties of quantum systems is defined. Several important classes of linear-time properties of quantum systems are examined, including safety, liveness, invariant and persistence properties. An algorithm for checking invariants of a quantum automaton is presented in Sec.~\ref{algo}. The techniques for model-checking safety properties and $\omega-$properties of quantum systems are presented in 
Sec.~\ref{reduce} and~\ref{reduce1}, respectively. A brief conclusion is drawn and some problems for future studies are pointed out in Sec.~\ref{concl} 

\section{Quantum Systems and Their Behaviors}\label{frame} 

\subsection{Hilbert Spaces}

According to a basic postulate of quantum mechanics, the state space
of an isolated quantum system is a Hilbert space. In this paper, we
only consider finite or countably infinite-dimensional Hilbert
spaces. For convenience of the reader, we briefly recall some basic
notions from Hilbert space theory. We write $\mathbf{C}$ for the set
of complex numbers. For each complex number $c\in \mathbf{C}$,
$\overline{c}$ stands for the conjugate of $c$. An inner product
over a complex vector space $H$ is a mapping
$\langle\cdot|\cdot\rangle: H\times H\rightarrow \mathbf{C}$
satisfying the following properties: \begin{enumerate}\item 
$\langle\varphi|\varphi\rangle\geq 0$ with equality if and only if
$|\varphi\rangle =0$; \item 
$\langle\varphi|\psi\rangle=\overline{\langle\psi|\varphi\rangle}$;
and \item $\langle\varphi|c_1\psi_1+c_2\psi_2\rangle=
c_1\langle\varphi|\psi_1\rangle+c_2\langle\varphi|\psi_2\rangle$ \end{enumerate} for
any $|\varphi\rangle, |\psi\rangle, |\psi_1\rangle, |\psi_2\rangle
\in H$ and for any $c_1,c_2\in \mathbf{C}$. 
Sometimes, we write $(|\varphi\rangle,|\psi\rangle)$ for the inner product $\langle\varphi|\psi\rangle$.
Two vectors
$|\varphi\rangle,|\psi\rangle$ in $H$ are said to be orthogonal and
we write $|\varphi\rangle\perp|\psi\rangle$ if
$\langle\varphi|\psi\rangle=0$.
For any vector $|\psi\rangle$ in
$H$, its length $||\psi||$ is defined to be
$\sqrt{\langle\psi|\psi\rangle}$. If $||\psi||=1$, then
$|\psi\rangle$ is called a unit vector. 

Let $H$ be an inner product
space, $\{|\psi_n\rangle\}$ a sequence of vectors in $H$, and
$|\psi\rangle\in H$. If for any $\epsilon
>0$, there exists a positive integer $N$ such that $||\psi_m
-\psi_n||<\epsilon$ for all $m,n\geq N$, then $\{|\psi_n\rangle\}$
is called a Cauchy sequence. If for any $\epsilon
>0$, there exists a positive integer $N$ such that $||\psi_n
-\psi||<\epsilon$ for all $n\geq N$, then $|\psi\rangle$ is called a
limit of $\{|\psi_n\rangle\}$ and we write $|\psi\rangle
=\lim_{n\rightarrow\infty}|\psi_n\rangle.$ A Hilbert space is a
complete inner product space; that is, an inner product space in
which each Cauchy sequence of vectors has a limit. A state of a
quantum system is usually described by a unit vector in a Hilbert
space. 

A sequence $\{|\psi_n\rangle\}$ of vectors in $H$
is summable with the sum $|\psi\rangle$ and we write
$|\psi\rangle=\sum_{n}|\psi_n\rangle$ if for any $\epsilon>0$ there
is nonnegative integer $n_0$ such that $$||\psi-\sum_{m\leq
n}\psi_m ||<\epsilon$$ for every $n\geq n_0$. A finite or countably
infinite family $\{|\psi_n\rangle\}$ of unit vectors is called an
orthonormal basis of $H$ if \begin{enumerate}\item $|\psi_m\rangle\perp |\psi_n\rangle$
for any $m,n$ with $m\neq n$; and \item 
$$|\psi\rangle=\sum_{n}\langle\psi_n|\psi\rangle|\psi_n\rangle$$ for
each $|\psi\rangle\in H.$\end{enumerate} 

Let $X\subseteq H$. If we have
$|\varphi\rangle +|\psi\rangle\in X$ and $c|\varphi\rangle \in X$
for any $|\varphi\rangle, |\psi\rangle\in X$ and $c\in \mathbf{C}$,
then $X$ is called a subspace of $H$. For each $X\subseteq H$, the
closure $\overline{X}$ of $X$ is defined to be the set of limits
$\lim_{n\rightarrow\infty}|\psi_n\rangle$ of sequences
$\{|\psi_n\rangle\}$ in $X$. A subspace $X$ of a Hilbert space $H$
is said to be closed if $\overline{X}= X$.  For any subset $X$ of $H$, we define $span X$ to be the smallest closed subspace of $H$.
Let $X$ be a closed
subspace of $H$ and $|\psi\rangle\in H$. 
Then we write $|\psi\rangle\perp X$ whenever $|\psi\rangle\perp |\varphi\rangle$ for all $|\varphi\rangle\in X$. The ortho-complementation of $X$ is defined to be $$X^{\bot}=\{|\varphi\rangle\in H||\varphi\rangle\perp X\}$$ For each $|\psi\rangle\in H$, there exist uniquely
$|\psi_0\rangle\in X$ and $|\psi_1\rangle\in
X^{\bot}$ such that
$|\psi\rangle=|\psi_0\rangle+|\psi_1\rangle$. The vector
$|\psi_0\rangle$ is called the projection of $|\psi\rangle$ onto $X$
and written $|\psi_0\rangle=P_X|\psi\rangle$. Thus, an operator
$P_X$ on $H$ is defined and it is called the projector onto $X$.

A (linear) operator on a Hilbert space $H$ is a mapping
$A: H\rightarrow H$ satisfying the following
conditions: \begin{enumerate}
\item $A(|\varphi\rangle+|\psi\rangle)=A|\varphi\rangle+A|\psi\rangle$;
\item $A(\lambda |\psi\rangle)=\lambda A|\psi\rangle$\end{enumerate} for all
$|\varphi\rangle,|\psi\rangle\in H$ and $\lambda\in\mathbf{C}$. The identity operator on $H$ is written as $I_H$. 
For any subset $X$ of $H$ and operator $A$ on $H$, the image of $X$ under $A$ is denoted by $$AX=\{A|\psi\rangle||\psi\rangle\in X\}.$$
For any operator $A$ on $H$, if there exists a linear
operator $A^{\dag}$ on $H$ such that $$(|\varphi\rangle,
A|\psi\rangle)=(A^{\dag}|\varphi\rangle,|\psi\rangle)$$ for all
$|\varphi\rangle, |\psi\rangle\in H$, then $A^{\dag}$ is
called the adjoint of $A$. An eigenvector of an operator $A$ on $H$ is a non-zero vector $|\psi\rangle\in H$ such that $A|\psi\rangle=\lambda |\psi\rangle$ for some $\lambda\in\mathbf{C}$, called the eigenvalue of $A$ corresponding to $|\psi\rangle$.

The state space of a composed quantum system is the tensor product of the state spaces of its component systems. Let $H_k$ be a Hilbert space with orthonormal basis $\{|\varphi_{i_k}\}$ for $1\leq k\leq n$. 
Then the tensor product $\bigotimes_{k=1}^{n}H_k$ is defined to be the Hilbert space with $\{|\varphi_{i_1}\rangle...|\varphi_{i_n}\rangle\}$ as its orthonormal basis. If $A_k$ is a linear operator 
on $H_k$ for $1\leq k\leq n$, then the tensor product $\bigotimes_{k=1}^{n}A_k$ is the operator on $\bigotimes_{k=1}^{n}H_k$ defined by $$\bigotimes_{k=1}^{n}A_k(|\psi_1\rangle...
|\psi_n\rangle)=(A_1|\psi_1\rangle)...(A_n|\psi_n\rangle)$$ for all $|\psi_k\rangle\in H_k$ $(1\leq k\leq n)$.

\subsection{Dynamics of Quantum Systems}
An operator $U$ on a Hilbert space $H$ is called a unitary transformation
if $U^{\dag}U=I_H$.
The basic postulate of quantum
mechanics about evolution of systems may be stated as follows:
Suppose that the states of a closed quantum system at times $t_0$
and $t$ are $|\psi_0\rangle$ and $|\psi\rangle$, respectively. Then
they are related to each other by a unitary operator $U$ which
depends only on the times $t_0$ and $t$:
$$|\psi\rangle=U|\psi_0\rangle.$$

\subsection{Quantum Automata}

As the first step toward to developing model-checking techniques for
quantum systems, we choose to consider a class of simple quantum
systems whose discrete-time behaviors can be modeled by quantum
automata~\cite{KW97}.

\begin{defn}Let $H$ be a Hilbert space with orthonormal basis $\{|i\rangle\}$. A quantum automaton in $H$ is a triple
$$\mathbb{A}=(Act,\{U_\alpha|\alpha\in Act\}, I)$$
where\begin{enumerate}\item $Act$ is a set of action names;
\item for each $\alpha\in Act$, $U_\alpha$ is a unitary operator on $H$, that is, it is a linear operator, written as $$U_\alpha|i\rangle=\sum_jU_\alpha(i,j)|j\rangle,$$ such that \begin{equation*}
\sum_j\overline{U_\alpha(j, i_1)}U_\alpha(i_2,j)=\begin{cases}1 &{\rm if}\ i_1=i_2,\\ 0 &{\rm otherwise;}
\end{cases}
\end{equation*} \item $I$ is a closed subspace of $H$, the space of initial
states.\end{enumerate}\end{defn}

A path of $\mathbb{A}$ is an infinite sequence
$|\psi_0\rangle|\psi_1\rangle|\psi_2\rangle...$ of states in $H$
such that $|\psi_0\rangle\in I$, and
$$|\psi_{n+1}\rangle=U_{\alpha_n}|\psi_n\rangle$$ for some $\alpha_n\in
Act$, for all $n\geq 0$. This means that a path starts in an initial
state $|\psi_0\rangle$, and for each $n\geq 0$, at the beginning of
the $n$th step the machine is in state $|\psi_n\rangle$, then it
performs an action described by $U_{\alpha_n}$ and evolves into state
$|\psi_{n+1}\rangle$. Likewise, a path fragment of $\mathbb{A}$ is a
finite sequence $|\psi_0\rangle|\psi_1\rangle...|\psi_n\rangle$ such
that $|\psi_0\rangle\in I$ and
$$|\psi_{k+1}\rangle=U_{\alpha_k}\psi_k\rangle$$ for some $\alpha_k\in
Act$, $k=0,1,...,n-1$. Let $|\psi\rangle\in I$ and let
$|\varphi\rangle$ be a state in $H$. We say that $|\varphi\rangle$ is
reachable from $|\psi\rangle$ in $\mathbb{A}$ if $\mathbb{A}$ has a
path fragment $|\psi_0\rangle|\psi_1\rangle...|\psi_n\rangle$ such
that $|\psi_0\rangle=|\psi\rangle$ and
$|\psi_n\rangle=|\varphi\rangle$.
We put $$R(\mathbb{A})=\{|\psi\rangle||\psi\rangle\ {\rm\ is\ reachable\ from\ some}\ |\phi\rangle\in I\}.$$ and define $RS(\mathbb{A})$ to be the closed subspace generated by $R(\mathbb{A})$, that is, $RS(\mathbb{A})=spanR(\mathbb{A})$. The following lemma gives a simple characterization of $RS(\mathbb{A})$. 

\begin{lem}\label{reachable}
$RS(\mathbb{A})$ is the intersection of all closed subspaces $X$ of $H$ satisfying the following conditions:\begin{enumerate}\item $I\subseteq X$; \item $U_{\alpha}X\subseteq X$ for all $\alpha\in Act$.\end{enumerate}
In other words, $RS(\mathbb{A})$ is the smallest one among all of these $X$.
\end{lem} 

\begin{proof} Straightforward.\end{proof}

\section{Linear-Time Properties of Quantum Systems}\label{ltp}

\subsection{Atomic Propositions in Quantum Systems}

Let $H$ be the state space of a quantum system. A closed subspace of
$H$ will be seen as an atomic proposition about this system; more
precisely, we will mainly consider the basic properties of the
system of the form: $|\psi\rangle\in X$, where $X$ is a closed
subspace of $H$, and $|\psi\rangle$ is a state of the system. So,
for a closed subspace $X$ of $H$, atomic proposition represented by
$X$ specifies a constraint on the behavior of the system under
consideration that its state is within the given region $X$. This
viewpoint of atomic propositions about a quantum system was proposed
by Birkhoff and von Neumann a long time ago, and it is exactly the starting point of their quantum logic~\cite{BvN36}. It was also adopted in one of the authors' studies on predicate 
transformer semantics~\cite{YDFJ10} and automata theory based on quantum logic~\cite{Y07}  

We write
$S(H)$ for the set of closed subspaces of $H$. Some basic (atomic)
propositions are of interest, but others may be irrelevant in a
special situation. So, we choose $AP\subseteq S(H)$. Intuitively, the
elements of $AP$ represents the atomic propositions of interest. For
each $|\psi\rangle\in H$, we write $L(|\psi\rangle)$ for the set of
atomic propositions satisfied in state $|\psi\rangle$; that is,
$$L(|\psi\rangle)=\{X\in AP||\psi\rangle\in X\}.$$

\begin{defn}Let $X\in S(H)$. Then we say that state $|\psi\rangle$ satisfies
$X$, written $|\psi\rangle\models X$, if $$\bigcap_{Y\in
L(|\psi\rangle)}Y\subseteq X.$$\end{defn}

Note that in the above definition $X$ is allowed to be not in $AP$.
The intuitive meaning of the inclusion in the above definition is
that the atomic propositions that hold in state $|\psi\rangle$ imply
collectively proposition $X$.

The following simple example provides a clear illustration of the above definition. 

\begin{exam}Let $H$ be an $n-$dimensional Hilbert space with orthonormal basis $\{|0\rangle,|1\rangle,...,|n-1\rangle\}$ $(n\geq 2)$, and let $|\psi\rangle=\frac{1}{\sqrt{2}}(|0\rangle+|1\rangle)$.\begin{enumerate}
\item If we take $AP=\{Y\in S(H)||0\rangle\perp Y\}$, then $L(|\psi\rangle)=\emptyset$ and $$\bigcap_{Y\in L(|\psi\rangle)}Y=H.$$ Thus, for any $X\in S(H)$, $|\psi\rangle\models X$ if and only if $X=H$.
\item Let $AP=\{2-{\rm dimensional\ subspaces\ of}\ H\}$. For the case of $n=2$, we have $$\bigcap_{Y\in L(|\psi\rangle)}Y=H,$$ and $|\psi\rangle\models X$ if and only if $X=H$. For the case of $n>2$,  
$$\bigcap_{Y\in L(|\psi\rangle)}Y=span\{|\psi\rangle\},$$ and $|\psi\rangle\models X$ if and only if $|\psi\rangle\in X$.
\item If $AP=\{X\in S(H)||2\rangle\in X\}$, then $$\bigcap_{Y\in L(|\psi\rangle)}Y=span\{|\psi\rangle, |2\rangle\},$$ and $|\psi\rangle\models X$ if and only if $|\psi\rangle, |2\rangle\in X$.
\end{enumerate}
\end{exam} 

We now present a technical lemma which will be frequently used in what follows. 
Recall that for
a finite family $\{X_i\}$ of closed subspaces of $H$, we define the
join of $\{X_i\}$ by $$\bigvee_{i}X_i=span (\bigcup_i X_i).$$ In
particular, we write $X\vee Y$ for the join of two closed subspaces
$X$ and $Y$ of $H$.

\begin{lem}\label{commut1} Suppose that $AP$ satisfies the
following two conditions:\begin{enumerate}\item Any two elements
$Z_1,Z_2$ of $AP$ commute; that is $P_{Z_1}P_{Z_2}=P_{Z_1}P_{Z_2}$, where $P_{Z_1},P_{Z_2}$ are
projections onto $Z_1$ and $Z_2$, respectively, and
\item $AP$ is closed under join: if $Z_1, Z_2\in AP$, then $Z_1\vee Z_2\in AP$.\end{enumerate}
Let $Y$ be a closed subspace of $H$ with $\{|\psi_i\rangle\}$ as its basis. Then the following two statements are equivalent:\begin{enumerate}\item $|\xi\rangle\models X$ for all $|\xi\rangle\in Y$; \item 
$|\psi_i\rangle\models X$ for all $i$.\end{enumerate}\end{lem} 

\begin{proof} It is obvious that 1) implies 2). Now we show that 2) implies 1). For any $|\xi\rangle\in Y$, we can write $$|\xi\rangle=\sum_{i\in I}a_i|\psi_i\rangle$$ for a finite index set $I$ and for some complex numbers $a_i$ $(i\in I)$ because $\{|\psi_i\rangle\}$ is a basis of $Y$. By the assumption that $|\psi_i\rangle\models X$, we obtain: $$\bigcap_{Z\in L(|\psi_i\rangle)}Z\subseteq X$$ for all $i\in I$. Therefore, it follows that $$\bigvee_{i\in I}\bigcap_{Z\in L(|\psi_i\rangle)}Z\subseteq X.$$ Since any two elements of $AP$ commute, distributivity is valid among $AP$ (see Proposition 2.5 in~\cite{BH00}), and we have: 
$$\bigvee_{i\in I}\bigcap_{Z\in L(|\psi_i\rangle)}Z=\bigcap_{\mathcal{Z}\in \prod_{i\in I}L(|\psi_i\rangle)}\bigvee_{i\in I}\mathcal{Z}(i).$$ Therefore, we only need to show that \begin{equation}\label{distri}
\bigcap_{Z\in L(|\psi\rangle)}Z\subseteq\bigcap_{\mathcal{Z}\in \prod_{i\in I}L(|\psi_i\rangle)}\bigvee_{i\in I}\mathcal{Z}(i).\end{equation} In fact, for any $$\mathcal{Z}\in \prod_{i\in I}L(|\psi_i\rangle),$$ by definition it holds that $|\psi_i\rangle\in\mathcal{Z}(i)$ for all $i\in I$. Then $$|\xi\rangle=\sum_{i\in I}a_i|\psi_i\rangle\in\bigvee_{i\in I}\mathcal{Z}(i).$$ In addition, it is assumed that $AP$ is closed under join. This implies $$\bigvee_{i\in I}\mathcal{Z}(i)\in L(|\psi\rangle)$$ and $$\bigcap_{Z\in L(|\psi\rangle)}Z\subseteq\bigvee_{i\in I}\mathcal{Z}(i).$$ So, Eq.~(\ref{distri}) is correct, and we complete the proof.  
\end{proof} 

\subsection{Linear-Time Properties and Satisfaction}

Now the set $AP$ of atomic propositions is fixed and we are going to define linear-time properties over $AP$.
We write $$(2^{AP})^{\ast}=\bigcup_{n=0}^{\infty}(2^{AP})^{n}$$ for
the set of finite sequences of subsets of $AP$ and
$(2^{AP})^{\omega}$ for the set of infinite sequences of subsets of
$AP$, where $\omega=\{0,1,2,...\}$ is the set of natural numbers. It what follows, we will use elements of $(2^{AP})^{\omega}$ (or $(2^{AP})^{\ast}$) to represent the behavior of a 
quantum system. This design decision deserves a careful explanation. Let $$\sigma=A_0A_1A_2...\in (2^{AP})^{\omega}\ ({\rm or}\ (2^{AP})^{\ast}).$$ Each element $A_n$ $(n\geq 0)$ is a closed subspace of 
the state space $H$ of a quantum system. So, it can be seen as a quantum object. However, if we do not care elements of $A_n$ $(n\geq 0)$ and focus our attention on $\sigma$ itself, then $\sigma$ is 
a classical object. Here, we can imagine that two levels exist in $\sigma$: object logical level and meta-logical level. The object logical level is the objects under consideration, so it belongs to the quantum world.
On the other hand, the meta-logical level is the way in which we (human beings) reason about the quantum world, so it is reasonably defined to be a classical object. In the sequel, we will see that the study of the behavior of a quantum system at the meta-logical level is similar to the classical case, but the study at the object logical level is very different because some essential differences between the quantum world and 
the classical world will irreversibly appear.    

For a path $\pi=|\psi_0\rangle |\psi_1\rangle|\psi_2\rangle...$ in a quantum automaton 
$\mathbb{A}$, we write
$$L(\pi)=L(|\psi_0\rangle)L(|\psi_1\rangle)L(|\psi_2\rangle)...\in
(2^{AP})^{\omega}.$$ Similarly, if
$\widehat{\pi}=|\psi_0\rangle|\psi_1\rangle...|\psi_n\rangle$ is a
path fragment in $\mathbb{A}$, then we write
$$L(\widehat{\pi})=L(|\psi_0\rangle)L(|\psi_1\rangle)...L(|\psi_n\rangle).$$

\begin{defn}The set of traces and the set of finite traces of a quantum
automaton $\mathbb{A}$ are defined as follows:
\begin{equation*}\begin{split}Traces(\mathbb{A})& =\{L(\pi)|\pi\ {\rm is\ a\ path\ in}\
\mathbb{A}\},\\
Traces_{fin}(\mathbb{A})&=\{L(\widehat{\pi})|\widehat{\pi}\ {\rm is\
a\ path\ fragment \ in}\
\mathbb{A}\}.\end{split}\end{equation*}
\end{defn}

Obviously, $Traces(\cdot)$ and $Traces_{fin}(\cdot)$
describes the infinite and finite behaviors, respectively, of quantum automatons.  
Note that what concerns us in this paper are only
linear-time behaviors of quantum systems since the behaviors of a
system is depicted in terms of sequences. In the future studies we
will also consider branching-time behaviors represented by trees
instead of sequences. But the branching-time behavior of a quantum system is much more complicated than its classical counterpart due to the superposition posibility of quantum states. 

A (linear-time) property of a quantum automaton $\mathbb{A}$ in
Hilbert space $H$ is then defined to be a subset $P$ of
$(2^{AP})^{\omega}$; in other words, an element of $P$ is an
infinite sequence $A_0A_1A_2...$ such that $A_n$ is a subset of $AP$
for all $n\geq 0$. A property $P$ specifies the admissible behaviors
of machine $\mathbb{A}$: if $A_0A_1A_2...\in P$, then a path
$\pi=|\psi_0\rangle|\psi_1\rangle|\psi_2\rangle...$ of $\mathbb{A}$
is admissible whenever $|\psi_n\rangle$ satisfies all the atomic
propositions in $A_n$ for all $n\geq 0$; otherwise the path $\pi$ is
prohibited by $P$.

Now we are ready to define the key notion of satisfaction of a property by a quantum system.

\begin{defn}We say that a quantum automaton $\mathbb{A}$ satisfies a
linear-time property $P$, written $\mathbb{A}\models P$, if
$Traces(\mathbb{A})\subseteq P$.\end{defn}

\subsection{Safety Properties}

In the remainder of this section, we consider several special classes of linear-time properties. Safety is one of the most important kinds of linear-time properties.  
A safety property specifies that \textquotedblleft something bad never
happens"~\cite{La77}. An elegant definition of safety property was introduced by Alpern and
Schneider~\cite{AS85} 
based on the intuition that a \textquotedblleft bad event" for a
safety property occurs in a finite amount of time, if it occurs at
all. Their definition can be naturally generalized to the quantum case by simply replacing atomic propositions about a classical system with closed subspaces of a Hilbert space. 
Formally, a finite sequence $\widehat{\sigma}\in
(2^{AP})^{\ast}$ is called a bad prefix of a property $P$ if
$\widehat{\sigma}\sigma\not\in P$ for all $\sigma\in
(2^{AP})^{\omega}$. We write $BPref(P)$ for the set of bad prefixes
of $P$. Let $\widehat{\sigma}\in (2^{AP})^{\ast}$ and $\sigma\in
(2^{AP})^{\omega}$. If $\sigma=\widehat{\sigma}\sigma^{\prime}$ for
some $\sigma^{\prime}\in (2^{AP})^{\omega}$, then $\widehat{\sigma}$
is said to be a prefix of $\sigma$.

\begin{defn}A property $P$ is called a safety property if any
$\sigma\not\in P$ has a prefix $\widehat{\sigma}\in
BPref(P).$\end{defn}

The following lemma gives a simple characterization of satisfaction
relation between quantum systems and safety properties.

\begin{lem}\label{bad}For any quantum automaton $\mathbb{A}$, and for
any safety property $P$, $\mathbb{A}\models P$ if and only if
$$Traces_{fin}(\mathbb{A})\cap BPref(P)=\emptyset.$$\end{lem}

For any $\widehat{\sigma}_1,\widehat{\sigma}_2\in (2^{AP})^{\ast}$,
if there is $\widehat{\sigma}\in (2^{AP})^{\ast}$ such that
$\widehat{\sigma}_1=\widehat{\sigma}_2\widehat{\sigma}$, then
$\widehat{\sigma}_2$ is called a prefix of $\widehat{\sigma}_1$ and
we write $\widehat{\sigma}_2\sqsubseteq \widehat{\sigma}_1$. It is
obvious that $\sqsubseteq$ is a partial order on $(2^{AP})^{\ast}$.
We write $MBPref(P)$ for the set of minimal bad prefixes of $P$,
that is, minimal elements of $BPref(P)$ 
according to order $\sqsubseteq.$ It is easy to see that
$BPref(P)$ in the definition of safety property and Lemma~\ref{bad}
can be replaced by $MBPref(P)$.

To conclude this subsection, we would like to point out that up to now our discussion on linear-time properties of quantum systems is almost the same as that for classical systems, e.g. 
the definition and characterization of safety property simply mimic their classical counterparts. However, some essential differences between classical and quantum systems will come out in the next subsection.  

\subsection{Invariants}

A special class of safety properties are invariants. Invariants will play a key role in the verification of safety properties for quantum systems. As in the classical case, the problem of 
model-checking a big class of safety properties will be reduced to the problem of checking invariants. 

\begin{defn}A property $P$ is said to be an invariant if there exists a
closed subspace $X$ of $H$ such that
\begin{equation}\label{inv}P=\{A_0A_1A_2...\in
(2^{AP})^{\omega}|\bigcap_{Y\in A_n}Y\subseteq X\ {\rm for\ all}\
n\geq 0\}.\end{equation}\end{defn}

Intuitively, the condition $$\bigcap_{Y\in A_n}Y\subseteq X$$ in
Eq.~(\ref{inv}) means that the atomic propositions in $A_n$ together
imply the proposition $X$. We will call $P$ the invariant defined by
$X$ and write $P=inv X$, and $X$ is often called the invariant condition of $invX$. 

As a concrete example, we consider
stabilizers~\cite{Go97}, which have been widely used in quantum error-correction (see for example~\cite{NC00}, Chapter 10) and measurement-based quantum computation~\cite{RBB03} as well as multi-partite teleportation and super-dense coding~\cite{WY08a, WY08b}.

\begin{exam} We write $$H_2=\{\alpha|0\rangle+\beta|1\rangle|\alpha,\beta\in\mathbf{C}\}$$ for the $2-$dimensional Hilbert space. So, $H_2$ 
is the state space of a single qubit. A state of a qubit is a vector $\alpha|0\rangle+\beta|1\rangle$ with $|\alpha|^{2}+|\beta|^{2}=1$. 
Let
$H=H_2^{\otimes n}$ be the tensor product of $n$ copies of $H_2$. Then it is the state space of $n$ qubits.
We write $I_2$ for the identity operator on $H_2$. The Pauli matrices $$X=\left(\begin{array}{cc}0 & 1\\1&0\end{array}\right),\ Y=\left(\begin{array}{cc}0 &-i\\ i&0\end{array}\right),\ Z=\left(\begin{array}{cc}1&0\\0&-1\end{array}\right)$$ are unitary operators on $H_2$. The set $$G_1=\{\pm I_2,\pm iI_2,\pm X,\pm iX,\pm Y,\pm iY,\pm Z,\pm iZ\}$$ forms a group with the composition of operators as its group operation. It is the Pauli group on a single qubit. More generally, the Pauli group on $n$ qubits is $$G_n=\{A_1\otimes ...\otimes A_n|A_1,...,A_n\in G_1\}$$  
Now let $S$ be a subgroup of
$G_n$ generated by $g_1,...,g_l$. Recall that a state
$|\psi\rangle\in H_2^{\otimes n}$ is stabilized by $S$ if
$g|\psi\rangle=|\psi\rangle$ for all $g\in S$. We put
$Act=\{\alpha_k:k=1,...,l\}$, $U_{\alpha_k}=g_k$ for $1\leq k\leq
l$, and $I=span\{|\psi\rangle\}$. Then
$$\mathbb{A}=(Act,\{U_\alpha|\alpha\in Act\},I)$$ is a quantum
automaton. Suppose that $AP$ contains all one-dimensional subspaces of $H$.
If $S$ is a stabilizer of $|\psi\rangle$, then 
$span\{|\psi\rangle\}$ is an invariant of $\mathbb{A}$, i.e.
$\mathbb{A}\models inv (span \{|\psi\rangle\}).$
Conversely, if $\mathbb{A}\models inv (span \{|\psi\rangle\})$, then 
$S$ is a stabilizer of $|\psi\rangle$ modulo phase shifts, i.e. 
for every $g\in S$, we have $g|\psi\rangle=e^{i\alpha}|\psi\rangle$ for some 
real number $\alpha$.
\end{exam}

Now we are going to give some conditions under which an invariant holds in a quantum automaton $\mathbb{A}=(Act,\{U_\alpha|\alpha\in Act\},I)$ with the state space $H$. 
First, we observe that $\mathbb{A}\models invX$ if and only if $
|\psi\rangle\models X$ for all states $|\psi\rangle\in R(\mathbb{A})$, that is, all states reachable from
some state $|\varphi\rangle\in I$. Note that the space $I$ of
initial states is a continuum. This is very different from the classical case where we usually only have finitely or countably infinitely many initial states. It will make that checking an invariant in a quantum system 
is much harder than that in a classical system. The following lemma shows that we only need to
consider the states reachable from a basis of $I$, which is a finite
set or at most a countably infinite set, under certain commutativity
of elements of $AP$ and closeness of $AP$ for join.

\begin{lem}\label{commut} Suppose that the initial states of quantum automaton $\mathbb{A}$ are
spanned by $\{|\psi_i\rangle\}$, that is,
$I=span\{|\psi_i\rangle\}$, and suppose that $AP$ is as in Lemma~\ref{commut1}. 
Then $\mathbb{A}\models invX$ if and only if $|\psi\rangle\models X$
for any state $|\psi\rangle$ reachable in $\mathbb{A}$ from some
$|\psi_i\rangle$, $i\geq 1$.
\end{lem}

\begin{proof} The \textquotedblleft only if" part is obvious. Now
we prove the \textquotedblleft if" part. It suffices to prove the following:\begin{itemize}\item\ Claim: If
$|\xi\rangle\models X$ for all state
$|\xi\rangle$ reachable from some $|\psi_i\rangle$, $i\geq 1$, then
$|\psi\rangle\models X$ for any state
$|\psi\rangle$ reachable from some state $|\varphi\rangle\in I$.\end{itemize} In fact, for
any $|\varphi\rangle\in I$, we can write
$$|\varphi\rangle=\sum_{i}a_i|\psi_i\rangle$$ for some complex numbers
$a_i$ because $I=span\{|\psi_i\rangle\}$. If $|\psi\rangle$ is
reachable from $|\varphi\rangle$, then there are
$\alpha_1,...,\alpha_n\in Act$, $n\geq 0$ such that
$$|\psi\rangle=U_{\alpha_n}...U_{\alpha_1}|\varphi\rangle.$$ We put
$$|\xi_i\rangle=U_{\alpha_n}...U_{\alpha_1}|\psi_i\rangle$$ for each
$i\geq 1$. Then $$|\psi\rangle=\sum_{i}a_i|\xi_i\rangle,$$ and
$|\xi_i\rangle$ is reachable from $|\psi_i\rangle$. It immediately follows from Lemma~\ref{commut1} that $|\psi\rangle\models X$ provided that $|\xi_i\rangle\models X$ for all $i\geq 1$. This completes the proof. 
\end{proof}

The above lemma will play a key role in the proofs of the main results in this paper (Theorems~\ref{main1} and~\ref{main2} below). It is worth mentioning again that both of them appeal to a certain commutativity of atomic propositions in $AP$. As is well-known,  non-commutativity of observables is one of the most essential features that distinguish quantum systems from classical systems. 
So, the commutativity condition in these lemmas is very restrictive. Fortunately, atomic propositions dealt with in these theorems just automatically enjoy the required commutativity.

The following
simple corollary gives a sufficient condition for invariant, which meets our intuition of invariant of a system very well.

\begin{cor}Suppose that $AP$ satisfies the two conditions in
Lemma~\ref{commut1}, and suppose that $I=span\{|\psi_i\rangle\}$. If
\begin{enumerate}\item $|\psi_i\rangle\models X$ for all $i$; and\item $U_\alpha Y\subseteq Y$ for all $Y\in AP$ and for
all $\alpha\in Act$,
\end{enumerate}then $\mathbb{A}\models inv X$.
\end{cor}

\begin{proof} We first have the following: \begin{itemize}\item \ Claim: 
$|\psi\rangle\models X$ implies $U_\alpha |\psi\rangle\models X$ for
all $\alpha\in Act$.\end{itemize} In fact, it follows from condition
2) that \begin{equation*}\begin{split}L(|\psi\rangle)&=\{Y\in AP||\psi\rangle\in Y\}\\ &\subseteq
\{Y\in AP||U_\alpha |\psi\rangle\in Y\}=L(U_\alpha |\psi\rangle).\end{split}\end{equation*}
Thus, if $|\psi\rangle\models X$, then
$$\bigcap_{Y\in L(U_\alpha |\psi\rangle)}Y\subseteq \bigcap_{Y\in L(|\psi\rangle)}Y\subseteq
X$$ and $U_\alpha |\psi\rangle\models X$. Now the proof is completed
by simply combining the above claim, condition (1) and
Lemma~\ref{commut}.\end{proof}

\subsection{Liveness Properties}

Liveness properties are another important kind of linear-time properties that are in a sense dual to safety properties. A liveness property specifies that \textquotedblleft something good will happen
eventually"~\cite{La77}. Alpern and Schneider's definition of liveness property~\cite{AS85} can be simply extended to quantum systems.

\begin{defn}A linear-time property $P\subseteq (2^{AP})^{\omega}$ is called a liveness property if for any $\widehat{\sigma}\in (2^{AP})^{\ast}$ there exists $\sigma\in (2^{AP})^{\omega}$ such that $\widehat{\sigma}\sigma\in P$.
\end{defn} 

 Some interesting characterizations of liveness properties (see~\cite{BK08}, Lemmas 3.35 and 3.38 and Theorem 3.37) can be easily generalized to the quantum case because their proofs are only based on the upper structure 
 of linear-time properties, which are entirely classical, and irrelevant to their bottom structure, namely the state spaces of quantum systems.  

Local unitary equivalence~\cite{Kr10} is a key criterion for classification of multipartite entanglements of which physicists are still far from a complete understanding.
The following example shows that local unitary equivalence can be properly described in
terms of liveness.

\begin{exam}Suppose that $H$ is a Hilbert space and $H^{\otimes n}$ is the tensor product of $n$ copies of $H$. Let $\mathcal{U}$ be a set of unitary operators on $H$. It is unnecessary that $\mathcal{U}$ contains all 
unitary operators on $H$. The elements of $\mathcal{U}$ can be understood as the operations allowed in the scenario under consideration. For any $U\in\mathcal{U}$ and $1\leq i\leq n$, $$U_i=U\otimes \bigotimes_{j\neq i}I_H$$ is a unitary operator on $H^{\otimes n}$ which performs $U$ on the $i$th copy of $H$ and does nothing on the other copies, where $I_H$ is the identity operator on $H$. So, $U_i$ can be seen as a local operation on $H^{\otimes n}$. For any two $n-$partite states $|\varphi\rangle,|\psi\rangle\in H^{\otimes n}$, if there exists a sequence $U^{(1)}_{i_1},...,U^{(m)}_{i_m}$ of local operations such that $$|\psi\rangle=U^{(1)}_{i_1}...U^{(m)}_{i_m}|\varphi\rangle,$$ then we say that $|\varphi\rangle$ and $|\psi\rangle$ are locally $\mathcal{U}-$equivalent. 

We can naturally construct a quantum automaton in $H^{\otimes n}$ that starts in state $|\varphi\rangle$ and performs local $\mathcal{U}-$operations: \begin{equation*}
\mathbb{A}=(H^{\otimes n}, Act=\{U_i|U\in\mathcal{U}\ {\rm and}\ 1\leq i\leq n\}, I=span\{|\varphi\rangle\})\end{equation*} Now we put $$P=\{A_0A_1A_2...\in (2^{AP})^{\omega}|\exists n\geq 0. A_n=span\{|\psi\rangle\}\}$$ Obviously, $P$ is a liveness property. It is easy to see that if $|\varphi\rangle$ and $|\psi\rangle$ are locally $\mathcal{U}-$equivalent, then $\mathbb{A}\models P$. Conversely, if $\mathbb{A}\models P$, then $|\varphi\rangle$ and $|\psi\rangle$ are locally $\mathcal{U}-$equivalent modulo phase shifts, i.e. $$|\psi\rangle=e^{i\alpha}U^{(1)}_{i_1}...U^{(m)}_{i_m}|\varphi\rangle$$ for some real number $\alpha$ and local operations $U^{(1)}_{i_1},...,U^{(m)}_{i_m}$.\end{exam}

\subsection{Persistence Properties}

Persistence properties are a very useful class of liveness properties. A persistence property asserts that a certain condition always holds from some moment on. 

\begin{defn}\label{d-pers}A property $P$ is called a persistence property if there
exists $X\in S(H)$ such that
\begin{equation}\label{pers}P=\{A_0A_1A_2...\in (2^{AP})^{\omega}|\exists
m.\forall n\geq m.\bigcap_{Y\in A_n}Y\subseteq X\}\end{equation}
\end{defn}

In this case that Eq.~(\ref{pers}) holds, we say that $P$ is the
persistence property defined by $X$ and write $P=pers X$.

As in the case of invariants, to check whether a persistence property is satisfied by a quantum automaton we have to consider the behaviors of the automaton starting in all initial states which form a continuum. The next lemma indicates that it suffices to consider the behavior starting in some basis states of the space of initial states if a certain commutativity is imposed on the atomic propositions in $AP$.

\begin{lem}\label{l-pers}Let $AP$ be as in Lemma~\ref{commut1}. Suppose that $I$ is
finite-dimensional and
$I=span\{|\psi_1\rangle,...,|\psi_k\rangle\}$. Then
$\mathbb{A}\models pers X$ if and only if for each $1\leq i\leq k$,
and for each path
$$|\psi_i\rangle=|\zeta_0\rangle\stackrel{U_{\alpha_0}}{\rightarrow}|\zeta_1\rangle\stackrel{U_{\alpha_1}}{\rightarrow}
|\zeta_2\rangle\stackrel{U_{\alpha_2}}{\rightarrow}...$$ starting in
a basis state $|\psi_i\rangle,$ there exists $m\geq 0$ such that
$|\zeta_n\rangle\models X$ for all $n\geq m$.
\end{lem}

\begin{proof} We only need to prove the \textquotedblleft if" part.
By Definition~\ref{d-pers} it suffices to show that for any path
$$|\eta_0\rangle \stackrel{U_{\alpha_0}}{\rightarrow}|\eta_1\rangle\stackrel{U_{\alpha_1}}{\rightarrow}
|\eta_2\rangle\stackrel{U_{\alpha_2}}{\rightarrow}...$$ in
$\mathbb{A}$, where $|\eta_0\rangle\in I$, we can find $m\geq 0$
such that $|\eta_n\rangle\models X$ for all
$n\geq m$.

Since $|\eta_0\rangle\in
I=span\{|\psi_1\rangle,...,|\psi_k\rangle\}$, we have
$$|\eta_0\rangle=\sum_{i=1}^{k}a_i|\psi_i\rangle$$ for some complex
numbers $a_i$ $(1\leq i\leq k)$. Put
$$|\zeta_{ij}\rangle=U_{\alpha_j}...U_{\alpha_1}U_{\alpha_0}|\psi_i\rangle$$
for all $1\leq i\leq k$ and $j\geq 0$. A simple calculation shows
that $$|\eta_j\rangle=\sum_{i=1}^{k}a_i|\zeta_{ij}\rangle$$ for all
$j\geq 0$. On the other hand, for each $1\leq i\leq k$, we have the
following transitions:
$$|\psi_i\rangle=|\zeta_{i0}\rangle\stackrel{U_{\alpha_0}}{\rightarrow}|\zeta_{i1}\rangle\stackrel{U_{\alpha_1}}{\rightarrow}
|\zeta_{i2}\rangle\stackrel{U_{\alpha_2}}{\rightarrow}...$$ in
$\mathbb{A}$. By the assumption, there is $m_i\geq 0$ such that
$|\zeta_{in}\rangle\models X$ for all $n\geq
m_i$. Let $m=\max_{i=1}^{k}m_i$. Then for all $n\geq m$, we have $|\zeta_{in}\rangle\models X$ for all $1\leq i\leq k$ and $$|\eta_n\rangle=\sum_{i=1}^{k}a_i|\zeta_{in}\rangle.$$ By Lemma~\ref{commut1} we obtain $|\eta_n\rangle\models X$ and thus complete the proof. 
\end{proof}

Note that except the conditions assumed in Lemma~\ref{commut1}, the above lemma also requires that the space of the initial states is finite-dimensional. This requirement is needed in the last step of the proof of the above lemma. 

Lemmas~\ref{commut} and~\ref{l-pers} will play a key role in the proofs of the main results in this paper (Theorems~\ref{main1} and~\ref{main2} below). It is worth mentioning again that both of them appeal to a certain commutativity of atomic propositions in $AP$. As is well-known,  non-commutativity of observables is one of the most essential features that distinguish quantum systems from classical systems. 
So, the commutativity condition in these lemmas is very restrictive. Fortunately, atomic propositions dealt with in these theorems just automatically enjoy the required commutativity.

The above lemma requires that the space of the initial states is finite-dimensional, but the Hilbert space $H$ can be infinite-dimensional. The following lemma indicates that persistence properties and invariants coincide whenever $H$ is finite-dimensional.  

\begin{lem}\label{inv-pers}
Suppose that $H$ is finite-dimensional and $AP$ is as in Lemma~\ref{commut1}. Then $\mathbb{A}\models persX$ if and only if $\mathbb{A}\models invX$.
\end{lem}
\begin{proof}The \textquotedblleft if" part is obvious. We now prove the \textquotedblleft only if" part. Assume that $\mathbb{A}\models pers X$ and we want to show that $\mathbb{A}\models inv X$. It suffices to demonstrate that $|\psi\rangle\models X$ for all $|\psi\rangle\in RS(\mathbb{A})$. Since $H$ is finite-dimensional, we can find a maximal set $\{|\psi_1\rangle,...,|\psi_l\rangle\}$ of linearly independent states in $RS(\mathbb{A})$. Then it should be a basis of $RS(\mathbb{A})$.

For each $1\leq i\leq l$, let $|\varphi_0\rangle|\varphi_1\rangle...|\varphi_n\rangle$ be a path fragment in $\mathbb{A}$ such that $|\varphi_0\rangle\in I$ and $|\varphi_n\rangle=|\psi_i\rangle$. We arbitrarily choose a unitary operator $U\in\{U_\alpha|\alpha\in Act\}$ and set $$|\varphi_{n+k}\rangle=U^{k}|\varphi_n\rangle$$ for all $k\geq 1$. Then the path fragment $|\varphi_0\rangle|\varphi_1\rangle...|\varphi_n\rangle$ is extended to a path $|\varphi_0\rangle|\varphi_1\rangle...|\varphi_n\rangle|\varphi_{n+1}\rangle...$ in $\mathbb{A}$. It follows from the assumption of $\mathbb{A}\models pers X$ that there exists $m_i\geq 0$ with $$U^{k}|\psi_i\rangle=|\varphi_{n+k}\rangle\models X$$ for all $k\geq m_i$. Put $m=\max_{i=1}^{l}m_i$. Then $U^{m}|\psi_i\rangle\models X$ for all $1\leq i\leq l$. By Lemma~\ref{commut1} we obtain that $|\psi\rangle\models X$ for all $$|\psi\rangle\in span\{U^{m}|\psi_i\rangle|1\leq i\leq l\}=U^{m}RS(\mathbb{A}).$$ By definition we have $UR(\mathbb{A})\subseteq R(\mathbb{A})$ and thus $URS(\mathbb{A})\subseteq RS(\mathbb{A})$. On the other hand, $\dim (URS(\mathbb{A}))=\dim (RS(\mathbb{A})$ because $U$ is a unitary operator. Then it follows that $URS(\mathbb{A})=RS(\mathbb{A})$. Consequently, it holds that $U^{m}RS(\mathbb{A})=RS(\mathbb{A})$, and we complete the proof.
\end{proof}
 
 We have a counterexample showing that the above lemma is not true in an infinite-dimensional Hilbert space $H$.  
 
 \begin{exam}Consider the space $l_2$ of square summable sequences: \begin{equation*}\begin{split}l_2=\{\sum_{n=-\infty}^{\infty}\alpha_n|n\rangle:\alpha_n\in\mathbf{C}\
{\rm for\ all}\ n\ {\rm and}\
\sum_{n=-\infty}^{\infty}|\alpha_n|^{2}<\infty\}.\end{split}\end{equation*} The inner product in
$l_2$ is defined by
$$(\sum_{n=-\infty}^{\infty}\alpha_n|n\rangle,\sum_{n=-\infty}^{\infty}\alpha^{\prime}|n\rangle)=\sum_{n=-\infty}^{\infty}
\alpha_n^{\ast}\alpha_n^{\prime}$$ for all
$\alpha_n,\alpha_n^{\prime}\in\mathbb{C}$, $-\infty <n<\infty$. The translation operator $U_{+}$ on
$l_2$ is defined by
$$U_{+}|n\rangle=|n+1\rangle$$ for all $n$. It is easy
to verify that $U_{+}$ is a unitary operator. Let $Act$ consist of a single action name $+$, $Act=\{+\}$, and $I=span\{|0\rangle\}$. Then $\mathbb{A}=(Act,\{U_\alpha|\alpha\in Act\},I)$ is a quantum automaton. Let $k$ be an integer, and let $$[k)=span\{|n\rangle|n\geq k\}$$ and $$(k-1]=span\{|n\rangle|n\leq k-1\}.$$ Put $AP=\{[k),(k-1], l_2\}$. Then $AP$ satisfies the two conditions in Lemma~\ref{commut1}. It is easy to see that $\mathbb{A}\models pers [k)$ but $\mathbb{A}\models inv [k)$ does not hold provided $k>0$. 
\end{exam}

\section{Algorithms for Checking Invariants}\label{algo}

In this section, we present an algorithm for checking invariants of a quantum automaton $\mathbb{A}=(Act,\{U_\alpha|\alpha\in Act\},I)$ in a finite-dimensional state space $H$. The design of this algorithm is based on Lemma~\ref{reachable} and the following observation: if $\{|\psi_1\rangle,...|\psi_l\rangle\}$ is a basis of $RS(\mathbb{A})$, then $\mathbb{A}\models inv X$ if and only if $|\psi_i\rangle\models X$ for all $1\leq i\leq l$. The last condition can be checked by a forward depth-first search.

\bigskip\

\textit{Algorithm:} Invariant checking.

\smallskip\

\textbf{Input}: \begin{enumerate}\item The set $\{U_\alpha|\alpha\in{Act}\}$ of the unitary operators in $\mathbb{A}$; \item A basis $\{|\psi_1\rangle,|\psi_2\rangle,...,|\psi_k\rangle\}$ of the space $I$ of initial states; \item A subspace $X$ of $H$.\end{enumerate}

\textbf{Output}: true if $\mathbb{A}\models inv X$, otherwise false.\\
{\bfseries set of} state $B := \phi$; (*a basis of $RS(\mathbb{A})$*)\\
{\bfseries stack of} state $S := \varepsilon$; (*the empty stack*)\\
{\bfseries bool} $b := \mathrm{true}$; (*all states in $B$ satisfy $X$*)\\
{\bfseries for} $i=1,2,\cdots,k$ {\bfseries do}\\
\indent$B := B \cup \{|\psi_i\rangle\}$; (*initial states are reachable*)\\
\indent$push(|\psi_i\rangle, S)$; (*start a depth-first search with initial states*)\\
\indent$b := b \wedge (|\psi_i\rangle\models X)$; (*check if all initial states satisfy $X$*)\\
{\bfseries od}\\
{\bfseries while} $(b \wedge S \neq \phi)$ {\bfseries do}\\
\indent$|\psi\rangle := top(S)$; (*consider a reachable state*)\\
\indent$pop(S)$;\\
\indent{\bfseries for all} $\alpha \in Act$ {\bfseries do}\\
\indent\indent$|\xi\rangle := U_\alpha|\psi\rangle$; (*get a candidate state*)\\
\indent\indent$b := b \wedge (|\xi\rangle \models X)$; (*check if $X$ is satisfied*)\\
\indent\indent{\bfseries if} $b \wedge |\xi\rangle\not\in spanB$ {\bfseries then} (*check if it has not been considered*)\\
\indent\indent\indent$B := B \cup \{|\xi\rangle\}$; (*extend $R$ by adding new reachable states*)\\
\indent\indent\indent$push(|\xi\rangle, S)$;\\
\indent\indent{\bfseries fi}\\
\indent{\bfseries od}\\
{\bfseries od}\\
{\bfseries return} $b$\\

\subsection{Analyzing the Algorithm} 

First, we observe that a candidate state $|\xi\rangle\in spanB$ would not be added into $B$. So the elements in $B$ are always linear independent, and thus there are at most $d=\dim H$ elements in $B$. Furthermore, note that a state would be pushed into $S$ if and only if it has been added into $B$. Then $S$ would become empty after popping at most $d$ states. This implies that the algorithm terminates after at most $d$ iterations of the {\bfseries while} loop.

Second, it is easy to check that all elements in $B$ are always reachable. In fact, the initial states $|\psi_i\rangle$ are reachable, and if some $|\psi\rangle\in B$ is reachable, then all candidate states $|\xi\rangle=U_\alpha |\psi\rangle$ are reachable. So, if an execution of the algorithm returns false, then there must be a reachable state $|\psi_i\rangle$ or some candidate state $|\xi\rangle$ that does not satisfy $X$.  

If the output is true, then according to Lemma~\ref{commut1}, all states in $B$, further in $spanB$, satisfy $X$. Therefore, the correctness of the above algorithm comes immediately from the following:

\begin{lem} $RS(\mathbb{A})\subseteq spanB$.\end{lem} 

\begin{proof} We only need to check that $spanB$ satisfies the conditions 1) and 2) in Lemma~\ref{reachable}. Condition 1) is satisfied as $|\psi_i\rangle\in B$ for all $1\leq i\leq k$. Note that for any $|\psi\rangle\in B$ and any $\alpha\in Act$, $U_\alpha |\psi\rangle$ was a candidate state at sometime, and then either $U_\alpha |\psi\rangle\in spanB$ or it would be added into $B$. So we always have $U_\alpha |\psi\rangle\in spanB$. Consequently, $$U_\alpha(spanB)=span(U_\alpha B)\subseteq span(spanB)=spanB$$ and condition 2) is also satisfied. \end{proof}

The algorithm is not feasible enough in practice although it has been proved to be theoretically correct as above. The reason is that different from the classical case where only a finite number of states are involved, the state space here is continuous, thus a state cannot be exactly record with a finite storage space. Then errors would be brought and accumulated during the excution, and make the result to be unstable. For example, the truth value of $|\xi\rangle\not\in spanB$ is quite sensitive to the error of $|\xi\rangle$ in our algorithm, so even a little error here may change this value and then change the excution of the algorithm a lot.

\subsection{Improving the Algorithm}

In this subsection, we show that the above algorithm can be dramatically improved whenever the unitary operator $U_\alpha$ has no degenerate eigenstates for every $\alpha\in Act$; more precisely, in this case, invariant checking of the quantum automaton $\mathbb{A}$ can be reduced to a problem of classical invariant checking.  

First, we observe that $RS(\mathbb{A})$ satisfies condition 2) in Lemma~\ref{reachable} and it can be rewritten as $U_\alpha RS(\mathbb{A})=RS(\mathbb{A})$, or equivalently, $$U_\alpha P_{RS(\mathbb{A})}=P_{RS(\mathbb{A})}U_\alpha,$$ where $P_{RS(\mathbb{A})}$ is the projection onto ${RS(\mathbb{A})}$, whenever $H$ is finite-dimensional. On the other hand, each $U_\alpha$ can be uniquely eigen-decomposed and thus has exactly $d$ eigenstates. Let $\lambda$ be an eigenvalue of $U_\alpha$ and $|\psi\rangle$ be the corresponding eigenstate. Then \begin{equation*}\begin{split}U_\alpha(P_{RS(\mathbb{A})} |\psi\rangle)&=(U_\alpha P_{RS(\mathbb{A})}) |\psi\rangle=(P_{RS(\mathbb{A})} U_\alpha) |\psi\rangle\\ &=P_{RS(\mathbb{A})}(U_\alpha |\psi\rangle)=\lambda(P_{RS(\mathbb{A})} |\psi\rangle).\end{split}\end{equation*} So, $P_{RS(\mathbb{A})} |\psi\rangle\propto |\psi\rangle$ and $${P_{RS(\mathbb{A})}}^{\perp} |\psi\rangle=|\psi\rangle-P_{RS(\mathbb{A})} |\psi\rangle\propto |\psi\rangle.$$ We have $P_{RS(\mathbb{A})} |\psi\rangle=0$ or ${P_{RS(\mathbb{A})}}^\perp |\psi\rangle=0$ since $$\langle\psi|P_{RS(\mathbb{A})}P_{RS(\mathbb{A})}^\perp|\psi\rangle=0.$$ Thus, every eigenstate of $U_\alpha$ should be in $RS(\mathbb{A})$ or in $RS(\mathbb{A})^\perp$. 

Recall that a transition systems is a $6-$tuple $$TS_C=(S_C,Act_C,\rightarrow_C,I_C,AP_C,L_C),$$ where \begin{enumerate}\item $S_C$ is a set of (classical) states; \item $Act_C$ is a set of the names of (classical) actions; \item $\rightarrow_C\subseteq S_C\times Act_C\times S_C$ is a transition relation; \item $I_C\subseteq S_C$ is a set of initial states; \item $AP_C$ is a set of (classical) atomic propositions; and \item $L_C:S_C\rightarrow 2^{AP_C}$ is a labeling function.\end{enumerate} We now construct a transition system $TS_C$ from the automaton $\mathbb{A}=(Act,\{U_\alpha|\alpha\in Act\},I)$ as follows:
\begin{enumerate}\item $S_C=\{\psi||\psi\rangle {\rm is\ an\ eigenstate\ of}\ U_{\alpha}\ {\rm for\ some}\ \alpha\in Act\}$, where each element $\psi$ in $S_C$ is regarded as the (classical) name of the corresponding quantum state $|\psi\rangle$; \item $Act_C=\{\tau\}$ consists of only one element $\tau$; \item $\rightarrow_C=\{(\psi,\tau,\phi)|\langle\psi|\phi\rangle\neq 0\}$; \item $I_C=\{\psi\in S_C||\psi\rangle\ {\rm is\ nonorthogonal\ to}\ I\}$; \item $AP_C=\{p_\psi|\ \psi\in S_C\}$, where for each $\psi\in S_C$, the atomic proposition $p_\psi$ is defined as follows: $\varphi\models p_\psi$ if and only if $\varphi=\psi$ for all $\varphi\in S_C$; and \item $L_C(\psi)=\{p_\psi\}$ for all $\psi\in S_C$.\end{enumerate} 

Next, for each closed subspace $X$ of Hilbert space $H$, we define a corresponding classical invariant property $P_{inv}$ over $AP_C$ as follows: 
$$P_{inv}=\{A_0A_1A_2...\in (2^{AP_C})^{\omega}|A_n\models \Phi\ {\rm for\ all}\ n\geq 0\}$$
where the invariant condition is 
$$\Phi=\underset{|\psi\rangle\models X}{\vee}p_\psi.$$ Furthermore, put $$R(TS_C)=\{\psi\in S_C|\psi\ {\rm is\ a\ reachable\ state\ of}\ TS_C\}.$$ Then we have: \begin{equation}\label{eqv}\begin{split}TS_C\models P_{inv}&\Leftrightarrow\forall\psi\in R(TS_C),\ \psi\models\Phi\\ &\Leftrightarrow\forall\psi\in R(TS_C),\ |\psi\rangle\models X.\end{split}\end{equation} Now we achieve our goal by showing the following:

\begin{lem} $\mathbb{A}\models invX$ if and only if $TS_C\models P_{inv}$.\end{lem}

\begin{proof} Let $RS(TS_C)$ be the subspace of $H$ spanned by the states $|\psi\rangle$ such that $\psi$ is reachable in $TS_C$; that is,  
$$RS(TS_C)=span\{|\psi\rangle|\psi\in R(TS_C)\}.$$
We have seen that $\mathbb{A}\models invX$ if and only if $|\psi\rangle\models X$ for all $|\psi\rangle\in RS(\mathbb{A})$.
Therefore, according to Lemma~\ref{commut1} and Eq.~(\ref{eqv}), we only need to show that $RS(\mathbb{A})=RS(TS)$. 

First, We demostrate that $RS(TS_C)\subseteq RS(\mathbb{A})$. If $|\psi\rangle\in R(TS_C)$, then $|\psi\rangle$ is a eigenstate of some $U_\alpha$ and thus is either in $RS(\mathbb{A})$ or in $RS(\mathbb{A})^\perp$. To show that 
$|\psi\rangle\in RS(\mathbb{A})$, 
we only need to prove that it is nonorthogonal to $RS(\mathbb{A})$. This can be done by an induction. For any $\psi\in I_C$, $|\psi\rangle$ is nonorthogonal to $I$ and thus is nonorthogonal to $RS(\mathbb{A})$. If $|\psi^{\prime}\rangle\in RS(\mathbb{A})$ and $\psi$ is a successor of $\psi^{\prime}$ in $TS_C$, then it holds that $\langle\psi^{\prime}|\psi\rangle\neq 0$, and $|\psi\rangle$ is nonorthogonal to $RS(\mathbb{A})$.

Second, we prove that $RS(\mathbb{A})\subseteq RS(TS_C)$. It suffices to verify that $RS(TS_C)$ satisfies the two conditions in Lemma~\ref{reachable}. We observe that $\langle\psi|\phi\rangle=0$ for any $\psi\in R(TS_C)$ and for any $\phi\in S_C\setminus R(TS_C)$, and $span\{|\psi\rangle|\psi\in S_C\}=H.$ Therefore, $$RS(TS_C)^\perp=span\{|\psi\rangle|\psi\in S_C\setminus R(TS_C)\}.$$ Notice that $\psi\perp I$ for all $\psi\in S_C\setminus R(TS_C)$. Thus, $I\perp RS(TS_C)^\perp$, and $I\subseteq RS(TS_C)$. So the condition 1) in Lemma~\ref{reachable} is satisfied. On the other hand, for any $\alpha\in Act$, assume that $|\psi_{\alpha 1}\rangle,|\psi_{\alpha 2}\rangle,\cdots,|\psi_{\alpha d}\rangle$ are the all eigenstates of $U_\alpha$, where the first $r$ states are in $R(TS_C)$ and the other $d-r$ ones are in $RS(TS_C)^\perp$. Since these $d$ states are pairwise orthogonal, we have $r\leq \dim RS(TS_C)$ and $$d-r\leq \dim RS(TS_C)^\perp=d-\dim RS(TS_C).$$ Thus, $r=\dim RS(TS_C)$. It means that $\{|\psi_{\alpha 1}\rangle,|\psi_{\alpha 2}\rangle,\cdots,|\psi_{\alpha r}\rangle\}$ is a basis of $RS(TS_C)$. Now, for any $|\psi\rangle\in RS(TS_C)$, let $$|\psi\rangle=\sum_{i\leq r}\mu_i|\psi_{\alpha i}\rangle.$$ We have $$U_\alpha|\psi\rangle=\sum_{i\leq r}\mu_i\lambda_{\alpha i}|\psi_{\alpha i}\rangle\in RS(TS_C),$$ where $\lambda_{\alpha i}$ is the corresponding eigenvalue of $|\psi_{\alpha i}\rangle$. Therefore, $U_\alpha RS(TS_C)\subseteq RS(TS_C)$, and the condition 2) in Lemma~\ref{reachable} is also satisfied.\end{proof}  

The above lemma allows us to adopt the algorithms for checking invariants of (classical) transition systems, e.g. Algorithms 3 and 4 presented in~\cite{BK08}, pages 109 and 110, to check invariants of quantum automata in which all unitary operators have no degenerate eigenstates.  

\section{Model Checking Reversible Safety Properties}\label{reduce}

One of the major techniques for verification of linear-time properties is automata-based model-checking~\cite{VW94, KV01}. This 
approach can reduce the problem of verifying a large class of linear-time properties to checking some specific properties for which algorithms are known. This section generalizes it to the quantum setting and establishes a reduction from verifying regular safety properties of quantum automata to checking their invariants, for which an algorithm was given in the last section. In this section and the next, we always assume that the Hilbert space $H$ is finite-dimensional. 

\subsection{Reversible Automata} 

The key idea of automata-based model-checking is to combine the system under consideration with an automaton that recognizes the property to be checked. Since the evolution of (closed) quantum systems is essentially reversible, it is reasonable to employ reversible automata in model-checking quantum systems.  

Recall that a nondeterministic finite automaton (an NFA for short) is a
quintuple
$$\mathcal{A}=(Q,\Sigma,\{\stackrel{A}{\rightarrow}|A\in\Sigma\},Q_0,F),$$
where $Q$ is a finite set of states, $\Sigma$ is an alphabet of
input symbols, $\stackrel{A}{\rightarrow}\ \subseteq Q\times Q$ is a
transition relation for each $A\in\Sigma$, $Q_0\subseteq Q$ is the
set of initial states, and $F\subseteq Q$ is the set of final
states. A word $w$ over alphabet $\Sigma$ is a finite string of
elements of $\Sigma$, i.e.
$$w\in\Sigma^{\ast}=\bigcup_{n=0}^{\infty}\Sigma^{n}.$$ A language
over $\Sigma$ is a subset of $\Sigma^{\ast}$. A word
$w=A_1A_2...A_n$ is accepted by $\mathcal{A}$ if there are $q_0\in
Q_0,q_1,...,q_{n-1}\in Q$ and $q_n\in F$ such that
$$q_0\stackrel{A_1}{\rightarrow}q_1\stackrel{A_2}{\rightarrow}...q_{n-1}\stackrel{A_n}{\rightarrow}q_n.$$
The language $L(\mathcal{A})$ accepted by $\mathcal{A}$ is defined
to be the set of the words accepted by $\mathcal{A}$. A language
over $\Sigma$ is called regular if it can be accepted by an NFA.

An NFA is called a deterministic finite automaton (DFA for short) if
$Q_0$ is a singleton and there are no pairs of transitions of the
form $q\stackrel{A}{\rightarrow}q_1$ and
$q\stackrel{A}{\rightarrow}q_2$ with $q_1\neq q_2$. Dually, an NFA
is said to be co-deterministic if $F$ is a singleton and there are
no pairs of transitions $q_1\stackrel{A}{\rightarrow}q$ and
$q_2\stackrel{A}{\rightarrow}q$ with $q_1\neq q_2$.

Reversible automata and the languages accepted by them have been thoroughly studied in~\cite{P87}, \cite{P01}. Here, we only recall the definition of reversible definition for convenience of the reader.

\begin{defn}An NFA
$\mathcal{A}=(Q,\Sigma,\{\stackrel{A}{\rightarrow}|A\in\Sigma\},Q_0,F)$
is said to be reversible if there are no pairs of
transitions of the form $q\stackrel{A}{\rightarrow}q_1$ and
$q\stackrel{A}{\rightarrow}q_2$ with $q_1\neq q_2$, and there are no
pairs of transitions $q_1\stackrel{A}{\rightarrow}q$ and
$q_2\stackrel{A}{\rightarrow}q$ with $q_1\neq q_2$.\end{defn}

\subsection{Products of Quantum Automata and Reversible Automata}

Let $\mathbb{A}=(Act,\{U_\alpha|\alpha\in Act\},I)$ be quantum
automaton in Hilbert space $H$. We can choose an orthonormal basis
of $I$ and then expand it to an orthonormal basis
$\{|\psi_i\rangle\}$ of $H$; in other words, we can choose an
orthnormal basis $\{|\psi_i\rangle\}$ of $H$ so that
$\{|\psi_i\rangle||\psi_i\rangle\in I\}$ is an orthnormal basis of
$I$. On the other hand, let $AP\subseteq S(H)$ be a finite set of
atomic propositions, and let $\Sigma=2^{AP}$. Suppose that
$\mathcal{A}=(Q,\Sigma,\{\stackrel{A}{\rightarrow}|A\in\Sigma\},Q_0,F)$
be a co-deterministic finite state automaton. It is asumed that
$Q_0\cap F=\emptyset$.
For each $A\in\Sigma=2^{AP}$ and for each $q\in Q$, we write
$$succ(q,A)=\{q^{\prime}\in Q|q\stackrel{A}{\rightarrow}q^{\prime}\
{\rm in}\ \mathcal{A}\}.$$ Then both $succ(q,A)=\emptyset$  and
$|succ(q,A)|\geq 1$ are possible. Whenever $succ(q,A)\neq\emptyset$,
we can choose an element $q_0^{\prime}\in succ(q,A)$. In particular,
for the case of $succ(q,A)\cap F\neq\emptyset$, we always choose
$q_0^{\prime}\in F$. Then we define $\delta(q,A)=q_0^{\prime}$. For
the case of $succ(q,A)=\emptyset$, $\delta(q,A)$ is undefined. Thus,
we define a partial function: $\delta:Q\times\Sigma\rightarrow Q$.

We write $$H_Q=span\{|q\rangle|q\in Q\}$$ for the Hilbert space with
$\{|q\rangle|q\in Q\}$ as its orthonormal basis. For each $n$, we
put $$Q_i=\{q\in Q|succ(q,L(|\psi_i\rangle))\neq\emptyset\}.$$ Since
$\mathcal{A}$ is co-deterministic, we have
$$|\{\delta(q,L(|\psi_i\rangle))|q\in Q_i\}|=|Q_i|.$$ Thus, there is a
bijection $$\kappa:Q\setminus Q_i\rightarrow
Q\setminus\{\delta(q,L(|\psi_i\rangle))|q\in Q_i\}.$$ For each
$\alpha\in Act$, we can define linear operator $V_\alpha$ on Hilbert
space $H\otimes H_Q$ as follows:
$$V_\alpha(|\psi_i\rangle|q\rangle)=\begin{cases}(U_\alpha
|\psi_i\rangle)|\delta(q,L(U_\alpha |\psi_i\rangle))\rangle\\ \ \ \ \ \ \ \ \ \ \ \ \ \ \ \ \ {\rm if}\ succ(q,L(|\psi_i\rangle))\neq\emptyset,\\
(U_\alpha |\psi_i\rangle)|\kappa(q)\rangle\ \ \ \ \ {\rm
otherwise}\end{cases}$$ for all $i$ and for all $q\in Q$. It is easy
to verify that $V_\alpha$ is a unitary operator by the assumption
that $\mathcal{A}$ is co-deterministic.

\begin{defn}\label{prod}The product of $\mathbb{A}$ and (a profile of) $\mathcal{A}$ is defined to be
the quantum automaton
$$\mathbb{A}\otimes\mathcal{A}=(Act,\{V_\alpha|\alpha\in
Act\},\mathbb{I})$$ in Hilbert space $H\otimes H_Q$, where
\begin{equation*}\begin{split}\mathbb{I}=span \{&|\psi_i\rangle|q\rangle|{\rm basis\ state}\ |\psi_i\rangle\in I\\ &{\rm
and}\ q_0\stackrel{L(|\psi_i\rangle)}{\rightarrow}q\ {\rm in}\
\mathcal{A}\ {\rm for\ some}\ q_0\in Q_0\}\end{split}
\end{equation*} is a closed subspace of $H\otimes H_Q$.\end{defn}

\subsection{Reversible Safety Properties}

Now let $P$ be a safety property over $AP$. Then the set $BPref(P)$
of bad prefixes of $P$ is a language over alphabet $\Sigma=2^{AP}$.
If it is a regular language, then $P$ is called a regular safety
property. For a regular safety property $P$, there exists an NFA
accepting $BPref(P)$. The subsets construction in automata theory
shows that $BPref(P)$ can be accepted by a DFA. By removing all
outgoing transitions from the final states we then obtain a DFA that
accepts $MBPref(P)$. So, $MBPref(P)$ is also a regular language over
alphabet $\Sigma=2^{AP}$. Furthermore, note that regular languages
are closed under reversal. So, there is also a co-deterministic
finite automaton $\mathcal{A}$ such that $L(\mathcal{A})=MBPref(P)$.
Note that for the case that the empty word is in $MBPref(P)$ we have
$P=\emptyset$. In what follows we simply exclude this trivial case.
Then it always holds that $Q_0\cap F=\emptyset$.

Our aim is to give a characterization of satisfaction relation
between quantum machines and regular safety properties in terms of
invariants. We choose the following set $\mathbb{AP}$ of atomic propositions on
$H\otimes H_Q$: $$\mathbb{AP}=\{H\otimes span\{|q\rangle|q\in
R\}|\emptyset\neq R\subseteq Q\}.$$ It is easy to see that $\mathbb{AP}$ satisfies the commutativity condition in Lemmas~\ref{commut1},~\ref{commut} and~\ref{l-pers}. The commutativity of $\mathbb{AP}$ is necessary for the main results in this section. First, we have the following:

\begin{prop}\label{half} Suppose that $P$ is a regular safety property and co-deterministic automaton $\mathcal{A}$ accepts
$MBPref(P)$. If $\mathbb{A}\models P$ then
\begin{equation}\label{prodx1}\mathbb{A}\otimes\mathcal{A}\models inv(H\otimes
span\{|q\rangle|q\in Q\setminus F\}).\end{equation}
\end{prop}

\begin{proof} It is easy to see that the set $\mathbb{AP}$ of
atomic propositions in $H\otimes H_Q$ satisfies conditions 1) and 2)
in Lemma~\ref{commut1}. We assume that $\mathbb{A}\models P$ and want
to show Eq.~(\ref{prodx1}). By Lemma~\ref{commut} and the
definition of $\mathbb{I}$ it suffices to show that for any basis
state $|\psi_i\rangle\in I$ and for any $q\in Q$ with
$$q_0\stackrel{L(|\psi_i\rangle)}{\rightarrow}q$$ for some $q\in Q_0$,
if $|\xi\rangle\in H\otimes H_Q$ is reachable from
$|\psi_i\rangle|q\rangle$, then $$|\xi\rangle\models H\otimes
span\{|q\rangle|q\in Q-F\}.$$ Suppose that
$$|\psi_i\rangle|q\rangle\stackrel{V_{\alpha_1}}{\rightarrow}|\xi_1\rangle\stackrel{V_{\alpha_2}}{\rightarrow}
...\stackrel{V_{\alpha_k}}{\rightarrow}|\xi_k\rangle=|\xi\rangle$$
for some $\alpha_1,\alpha_2,...,\alpha_k\in Act$. By the definition
of $V_\alpha$'s we obtain:\begin{equation*}\begin{split}
&|\xi_1\rangle=(U_{\alpha_1}|\psi_i\rangle)|q_1\rangle,\\ & \ \ \ \ \ \ \ \ \ \ \ \ \ \ \
\ \ \ \ \ \ \ \ \ \ \ \ \ \ \ \ \ \ \ \ \ \ \ 
q\stackrel{L(U_{\alpha_1}|\psi_i\rangle)}{\rightarrow}|q_1\rangle,\\
&|\xi_2\rangle=(U_{\alpha_2}U_{\alpha_1}|\psi_i\rangle)|q_2\rangle,\\ & \ \ \ \ \ \ \ \ \ \ \ \ \ \ \
\ \ \ \ \ \ \ \ \ \ \ \ \ \ \ \ \ \ \ \ \ \ \ 
q_1\stackrel{L(U_{\alpha_2}U_{\alpha_1}|\psi_i\rangle)}{\rightarrow}|q_2\rangle, \\
&\ \ \ \ \ \ \ \ \ \ \ \ \ \ \ \ \ \ \ \ \cdots \cdots \cdots \cdots\\
&|\xi_k\rangle=(U_{\alpha_k}...U_{\alpha_1}|\psi_i\rangle)|q_2\rangle,\\ & \ \ \ \ \ \ \ \ \ \ \ \ \ \ \
\ \ \ \ \ \ \ \ \ \ \ \ \ \ \ \ \ \ \ \ \ \ \ 
q_{k-1}\stackrel{L(U_{\alpha_k}...U_{\alpha_1}|\psi_i\rangle)}{\rightarrow}|q_k\rangle.
\end{split}\end{equation*}Then we have: 
$$\widehat{\pi}=|\psi_i\rangle(U_{\alpha_1}|\psi_i\rangle)...(U_{\alpha_k}...U_{\alpha_1}|\psi_i\rangle)$$
is a path fragment in $\mathbb{A}$. Since $|\psi_i\rangle\in I$, we
obtain:
$$\widehat{\sigma}=L(|\psi_i\rangle)L(U_{\alpha_1|\psi_i\rangle})...L(U_{\alpha_k}...U_{\alpha_1}|\psi_i\rangle)\in
Traces_{fin}(\mathbb{A}).$$ It follows from Lemma~\ref{bad} that
$$Traces_{fin}(\mathbb{A})\cap MBPref(P)=\emptyset$$ because
$\mathbb{A}\models P$. Thus, $\widehat{\sigma}\not\in
MBPref(P)=L(\mathcal{A})$ and $q_k\not\in F$. Consequently, it holds
that
$$|\xi_k\rangle=(U_{\alpha_k}...U_{\alpha_1}|\psi_i\rangle)|q_k\rangle\in
H\otimes span\{|q\rangle|q\in Q\setminus F\}\in \mathbb{AP}$$ and
$$|\xi\rangle=|\xi_k\rangle\models H\otimes span\{|q\rangle|q\in
Q\setminus F\}.$$\end{proof}

It is easy to see that in general the inverse of the above
proposition is incorrect. However, it holds for the safety
properties whose bad prefixes accepted by reversible
automata~\cite{P01}. 

\begin{defn}A safety property $P$ is said to be reversible
if $MBPref(P)$ is accepted by a reversible automaton $\mathcal{A}$.\end{defn}

Now we are ready to present one of the main results in this paper.

\begin{thm}\label{main1}If $P$ is a reversible safety property and $\mathcal{A}$ a reversible automaton with
$\mathcal{L}(\mathcal{A})=MBPref(P)$, then $\mathbb{A}\models P$ if and only if
\begin{equation}\label{prodx2}\mathbb{A}\otimes\mathcal{A}\models inv (H\otimes
span\{|q\rangle|q\in Q\setminus F\}).\end{equation}\end{thm}

\begin{proof} With Proposition~\ref{half}, we only need to show
that Eq.~(\ref{prodx2}) implies $\mathbb{A}\models P$.
This can be done by refutation. If $\mathbb{A}\models P$ does not
hold, then it follows from Lemma~\ref{bad} that
$$Traces_{fin}(\mathbb{A})\cap MBPref(P)\neq\emptyset.$$ Then there is
a path fragment
$\widehat{\pi}=|\psi_0\rangle|\varphi_1\rangle...|\varphi_k\rangle$
in $\mathbb{A}$ such that $|\varphi_0\rangle\in I$ and
\begin{equation*}\begin{split}\widehat{\sigma}=L(\widehat{\pi})=L(|\psi_0\rangle)L(|\varphi_1\rangle)...L(|\varphi_n\rangle)\in
MBPref(P)=L(\mathcal{A}).\end{split}\end{equation*} First, there are
$\alpha_1,...,\alpha_n\in Act$ such that
$|\varphi_{j+1}\rangle=U_{\alpha_{j+1}}|\varphi_j\rangle$ for
$j=0,1,...,n-1$. Secondly, by definition there are
$q_{-1},q_0,q_1,...,q_n\in Q$ such that $q_{-1}\in Q_0$, $q_n\in F$
and the transitions
$$q_{-1}\stackrel{L(|\varphi_0\rangle)}{\rightarrow}q_0\stackrel{L(|\varphi_1\rangle)}{\rightarrow}q_1\cdot\cdot\cdot
\stackrel{L(|\varphi_n\rangle)}{\rightarrow}q_n$$ hold in
$\mathcal{A}$. Since $\mathcal{A}$ is reversible, we obtain:
$$q_{j+1}=\delta(q_j,L(|\psi_{j+1}\rangle))$$ for $j=-1,0,1,...,n$.
Therefore, $|\varphi_0\rangle|q_0\rangle\in\mathbb{I}$ and we have
$$|\varphi_0\rangle|q_0\rangle\stackrel{V_{\alpha_1}}{\rightarrow}|\varphi_1
\rangle|q_1\rangle\stackrel{V_{\alpha_2}}{\rightarrow}...\stackrel{V_{\alpha_{n-1}}}{\rightarrow}
|\varphi_{n-1}\rangle|q_{n-1}\rangle
\stackrel{V_{\alpha_n}}{\rightarrow}|\varphi_n\rangle|q_n\rangle$$
in $\mathbb{A}\otimes\mathcal{A}$. So,
$|\varphi_n\rangle|q_n\rangle$ is reachable from
$|\varphi_0\rangle|q_0\rangle$. However,
$$L(|\varphi_n\rangle|q_n\rangle)=\{H\otimes span\{|q\rangle|q\in
R\}|q_n\in R\subseteq Q\}$$ and
\begin{equation*}\begin{split}\bigcap_{Y\in
L(|\varphi_n\rangle|q_n\rangle)}Y&=H\otimes span\{|q_n\rangle\}\\
&\not\subseteq H\otimes span\{|q\rangle|q\in Q\setminus
F\}\end{split}\end{equation*} because $q_k\in F$. This means that
$$|\varphi_n\rangle|q_n\rangle\not\models H\otimes
span\{|q\rangle|q\in Q\setminus F\}.$$ Consequently,
$$\mathbb{A}\otimes\mathcal{A}\not\models inv(H\otimes
span\{|q\rangle|q\in Q\setminus F\}).$$\end{proof}

The above theorem reduces the problem of checking a reversible safety property for the quantum automaton $\mathbb{A}$ to checking an invariant for the quantum automaton $\mathbb{A}\otimes\mathcal{A}$, for which an algorithms was already given in Sec.~\ref{algo}. 

\section{Model-Checking $\omega$-Reversible Properties}\label{reduce1}

The results given in the last section can be generalized to a larger class of linear-time properties by using reversible B\"uchi automata. A B\"uchi automaton is an NFA accepting infinite words. Let $\mathcal{A}=(Q,\Sigma,\{\stackrel{A}{\rightarrow}|A\in\Sigma\},Q_0,F)$ be an NFA. We write $\Sigma^{\omega}$ for the set of $\omega-$words over $\Sigma$, i.e. infinite sequences of elements of $\Sigma$. An $\omega-$word $w=A_0A_1A_2...\in\Sigma^{\omega}$ is accepted by B\"uchi automaton $\mathcal{A}$ if there exists an infinite sequence $q_0,q_1,q_2,...$ in $Q$ such that $q_0\in Q_0$, $$q_0\stackrel{A_0}{\rightarrow}q_1\stackrel{A_1}{\rightarrow}q_2\stackrel{A_2}{\rightarrow}...$$ and $q_n\in F$ for infinitely many $n\geq 0$. The language $\mathcal{L}_{\omega}(\mathcal{A})$ accepted by B\"uchi automaton $\mathcal{A}$ is defined to be the set of $\omega-$words accepted by $\mathcal{A}$.

First, Proposition~\ref{half} can be generalized as follows. 

\begin{prop}\label{om-r}Let $P$ be a linear-time property and $\mathcal{A}$ a
co-deterministic finite state automaton such that
$\mathcal{L}_{\omega}(\mathcal{A})=(2^{AP})^{\omega}\setminus P$.
Then $\mathbb{A}\models P$ implies
$$\mathbb{A}\otimes\mathcal{A}\models pers(H\otimes
span\{|q\rangle|q\in Q\setminus F\}).$$
\end{prop}

\begin{proof} By Lemma~\ref{l-pers}, it suffices to show that for
any path
$$|\psi_i\rangle|q\rangle\stackrel{V_{\alpha_1}}{\rightarrow}|\zeta_1\rangle\stackrel{V_{\alpha_2}}{\rightarrow}|\zeta_2\rangle
\stackrel{V_{\alpha_3}}{\rightarrow}...$$ where $|\psi_i\rangle$ is
a basis state of $I$,
$$q_0\stackrel{L(|\psi_i\rangle)}{\rightarrow}q$$ and $q_0\in Q_0$,
there exists $m\geq 0$ such that $$|\zeta_n\rangle\models H\otimes
span\{|q\rangle|q\in Q\setminus F\}.$$ We write: 
$$|\varphi_n\rangle=U_{\alpha_n}...U_{\alpha_1}|\psi_i\rangle$$ for
all $n\geq 1$. By the definition of $V_\alpha$'s we have
$|\zeta_n\rangle=|\varphi_n\rangle|q_n\rangle$ for all $n\geq 1$,
and
\begin{equation}\label{eq1}|\psi_i\rangle\stackrel{U_{\alpha_1}}{\rightarrow}|\varphi_1\rangle\stackrel{U_{\alpha_2}}{\rightarrow}|\varphi_2\rangle
\stackrel{U_{\alpha_3}}{\rightarrow}...\end{equation}
\begin{equation}\label{eq2}q\stackrel{L(|\psi_i\rangle)}{\rightarrow}q_1\stackrel{L(|\varphi_1\rangle)}{\rightarrow}
q_2\stackrel{L(|\varphi_2\rangle)}{\rightarrow}q_3\stackrel{L(|\varphi_3\rangle)}{\rightarrow}...
\end{equation} Therefore, it follows from Eq.~(\ref{eq1}) that
$$\sigma=L(|\psi_i\rangle)L(|\varphi_1\rangle)L(|\varphi_2\rangle)...\in
Traces(\mathbb{A})\subseteq P$$ and $$\sigma\notin
(2^{AP})^{\omega}\setminus P=\mathcal{L}_\omega(\mathcal{A}).$$ This
together with Eq.~(\ref{eq2}) implies that there is $m\geq 0$ such
that for $n\geq m$, we have $q_n\in Q\setminus F$, i.e.
$$|\zeta_n\rangle=|\varphi_n\rangle|q_n\rangle\in H\otimes
span\{|q\rangle|q\in Q\setminus F\}.$$\end{proof}

As in the case of safety properties, the inverse of the above proposition requires that the B\"uchi automaton accepting property $P$ is reversible. So, we have the following generalization of Theorem~\ref{main1}.

\begin{thm}\label{main2}If $P$ is a linear-time property and
$\mathcal{A}$ a reversible automaton with
$\mathcal{L}_\omega(\mathcal{A})=(2^{AP})^{\omega}\setminus P$, then
$\mathbb{A}\models P$ if and only if
\begin{equation}\label{prodx3}\mathbb{A}\otimes\mathcal{A}\models pers(H\otimes
span\{|q\rangle|q\in Q\setminus F\}).\end{equation}\end{thm}

\begin{proof} The \textquotedblleft only if" part is exactly
Proposition~\ref{om-r}. For the \textquotedblleft if" part, assume
that Eq.~(\ref{prodx3}) is correct. We aim at proving
$\mathbb{A}\models P$ by refutation. If $\mathbb{A}\not\models P$,
then there exists a path
$$|\varphi_0\rangle\stackrel{U_{\alpha_0}}{\rightarrow}|\varphi_1\rangle\stackrel{U_{\alpha_1}}{\rightarrow}
|\varphi_2\rangle\stackrel{U_{\alpha_2}}{\rightarrow}...$$ in
$\mathbb{A}$ such that $|\varphi_0\rangle\in I$ and
$$L(|\varphi_0\rangle)L(|\varphi_1\rangle)L(|\varphi_2\rangle)...\in
(2^{AP})^{\omega}\setminus P=\mathcal{L}_\omega(\mathcal{A}).$$
Consequently, we have a path
$$q_{-1}\stackrel{L(|\varphi_0\rangle}{\rightarrow}q_0\stackrel{L(|\varphi_1\rangle}{\rightarrow}q_1
\stackrel{L(|\varphi_2\rangle}{\rightarrow}...$$ in $\mathcal{A}$
such that $q_{-1}\in Q_0$ and $q_j\in F$ for infinitely many $i$.
The assumption that $\mathcal{A}$ is reversible implies that
$$q_{j+1}=\delta(q_j,L(|\varphi_{j+1}\rangle))$$ for all $j\geq -1$.
Thus, by Definition~\ref{prod} we obtain a path
$$|\varphi_0\rangle|q_0\rangle\stackrel{V_{\alpha_0}}{\rightarrow}
|\varphi_1\rangle|q_1\rangle\stackrel{V_{\alpha_1}}{\rightarrow}
|\varphi_2\rangle|q_2\rangle\stackrel{V_{\alpha_2}}{\rightarrow}...
$$ in $\mathbb{A}\otimes\mathcal{A}$ with
$|\varphi_0\rangle|q_0\rangle\in\mathbb{I}$, but
\begin{equation*}\begin{split}L(|\varphi_0\rangle|q_0\rangle)L(|\varphi_1\rangle|q_1\rangle)
&L(|\varphi_2\rangle|q_2\rangle)...\\ &\notin pers(H\otimes
span\{|q\rangle|q\in Q\setminus F\}\end{split}\end{equation*} since
$$|\varphi_j\rangle|q_j\rangle\not\models H\otimes
span\{|q\rangle|q\in Q\setminus F\}$$ for infinitely many $j$. This
is a contradiction.\end{proof}

By the above theorem, we are able to reduce the problem of checking an $\omega-$reversible property of the quantum automaton $\mathbb{A}$ to checking a persistence property of quantum automaton $\mathbb{A}\otimes\mathcal{A}$, which can be further reduced to checking an invariant by using Lemma~\ref{inv-pers}. Therefore, the problem of checking $\omega-$reversible properties of quantum systems can be eventually solved by employing the algorithm presented in Sec.~\ref{algo}.

\section{Conclusion}\label{concl}

This paper aims at developing effective techniques for model-checking linear-time properties of quantum systems. It can be seen as one of the first steps toward to a theoretical foundation for (classical) computer-aided verification of engineered quantum systems.   
The main contribution of the paper includes:\begin{itemize}\item We define a mathematical framework in which we can
examine various linear-time properties of quantum systems, such as safety and liveness properties. 
\item We present an algorithm for checking invariants of quantum systems. 
\item We show that both checking a safety
property of a (closed) quantum system recognizable by a reversible automaton and checking a linear-time property of a (closed) quantum system recognizable by a reversible B\"uchi automaton 
can be done by verifying an invariant of a larger system.\end{itemize}

The physical implication of the automata-based approach to model-checking a quantum system is very interesting. There are two systems involved in this approach. One of them is the quantum system $\mathbb{A}$ to be checked. It can be called the object system, and we assume that its state space is $H$. The other system is a classical system whose behavior is described by an automaton $\mathcal{A}$. We call it the probe system. The object system and the probe system then interact to form the system $\mathbb{A}\otimes\mathcal{A}$. The automaton-based approach allows us to check a property of the object system by means of checking an invariant of $\mathbb{A}\otimes\mathcal{A}$. Note that the invariant condition needed to be checked is of the form $H\otimes X$, where $X$ is a subspace of the state space of the probe system (see Theorems~\ref{main1} and~\ref{main2}). So, only the probe system will be examined in checking such an invariant. Obviously, the idea of automata-based model-checking coincides with that of indirect quantum measurements (see for example~\cite{BP02}, Sec. 2.4.6). This interesting physical meaning of automata-based approach have been overlooked in the classical case. In the quantum case, it is even more interesting to notice that the probe system is a classical system, and thus the problem of checking a quantum system is reduced to checking a classical system.  

As is well-known, the most serious disadvantage of model-checking is
the state explosion problem. This problem should not be very serious
in the early time of applying model-checking techniques to quantum
engineering. As one can imagine, the size of quantum engineering
systems that will be implemented in the near future cannot be very
large. On the other hand, the errors in the design of these systems
will not be caused mainly by their large sizes that the designers
are unable to manage. Instead, they may be caused by the anti-human
intuition features of the quantum world that the designers cannot
properly understand. So, we believe that model-checking techniques
based on a solid mathematical model of quantum systems will be vital
in guaranteeing correctness and safety of quantum engineering
systems.

The results achieved in this paper are only a very small step toward to the general purpose of model-checking quantum systems, and a lot of important problems are still unsolved. Here, we would like to mention a few open problems for further studies: \begin{itemize}\item \textit{Non-probabilistic vs probabilistic (atomic) propositions:}  Only non-probabistic atomic propositions are considered in this paper, following the basic idea of Birkhoff-von Neumann quantum logic~\cite{BvN36}. However, quantum mechanics is essentially a statistical theory based on quantum measurements. So, more sophisticated model-checking techniques for quantum systems should be able to encompass probabilistic information through incorporating checking with the theory of quantum measurements.
\item \textit{Closed vs open quantum systems:} In this paper, quantum systems are modeled by quantum automata whose behaviors are described by unitary operators. According a basic postulate of quantum mechanics, 
unitary operators are suited to depict the dynamics of closed quantum systems. A more suitable mathematical formalism for evolution of open quantum systems that interact with the environment 
is given in terms of super-operators~\cite{NC00} (see chapter 8). So, an interesting topic for further studies is to extend the model-checking technique developed in this paper so that it can be applied to quantum systems modeled by quantum automata with super-operators as their description of transitions.    
\item \textit{Linear-time vs branching-time:}
The algorithms presented in this paper can only check linear-time properties of quantum systems. One may naturally expect to develop model-checking techniques for quantum systems that can verify branching-time properties. The first step toward such an objective would be to define a logic that can properly specify branching-time behaviors of quantum systems. A quantum extension of computation tree logic was already proposed by Baltazar, Chadha, Mateus and Sernadas~\cite{BCMS07}, \cite{BCM08}. It seems that more research in this direction is in order because the branching notion of time for quantum systems is highly related to some foundational problems of quantum mechanics, e.g. trajectories~\cite{Br02}, decoherent (or consistent) histories~\cite{Gr96}, that are still not well-understood even in the physicists community.    
\item \textit{Classical vs quantum algorithms:} The algorithms for model-checking quantum systems developed in this
paper are classical. As the progress of quantum engineering, more
and more complicated quantum systems will be produced, and classical
algorithms might be too slow for checking their correctness and
safety. But the development of quantum engineering might make that
large-scalable and functional quantum computers be eventually built,
and quantum computer will be widely used in quantum engineering just
as today's computers are used in today's engineering. An interesting
open problem would be to design quantum algorithms for
model-checking quantum systems (as well as classical systems).\end{itemize}

\end{document}